\documentclass[11pt]{amsart}
\usepackage{amsthm}
\usepackage{amsmath,amstext,amsthm,amsfonts,amssymb,amsthm}
\usepackage{mathrsfs}
\usepackage{multicol}
\usepackage{latexsym}
\usepackage{color}
\usepackage{graphicx}
\usepackage{epsf}
\usepackage{epsfig}
\usepackage{epic}
\usepackage{graphics}
\usepackage{verbatim}
\usepackage{color}
\usepackage{hyperref}
\usepackage{bm}
\newtheorem{theorem}{Theorem}[section]
\newtheorem{lemma}[theorem]{Lemma}
\newtheorem{corollary}[theorem]{Corollary}
\theoremstyle{definition}
\newtheorem{definition}[theorem]{Definition}
\newtheorem{proposition}[theorem]{Proposition}

\theoremstyle{remark}
\newtheorem{remark}[theorem]{Remark}
\numberwithin{equation}{section}

\newcommand{\B}{\mathcal{B}}

\newcommand{\quat}{\mathbb H}

\newcommand{\mc}{\mathcal}

\newcommand{\be}{\begin{equation}}
\newcommand{\en}{\end{equation}}
\newcommand{\D}{{\mc D}}

\newcommand{\bedefin}{\begin{defi}}
	\newcommand{\findefi}{\end{defi} \medskip}
\newcommand{\betheo}{\begin{theorem}$\!\!${\bf \,\,\,}}
	\newcommand{\entheo}{\end{theorem}}
\newcommand{\enth}{\end{theorem}}
\newcommand{\becor}{\begin{cor}$\!\!${\bf .}}
	\newcommand{\encor}{\end{cor}}
\newcommand{\belem}{\begin{lem}$\!\!${\bf .}}
	\newcommand{\enlem}{\end{lem}}
\newcommand{\bea}{\begin{eqnarray}}
\newcommand{\ena}{\end{eqnarray}}
\newcommand{\beano}{\begin{eqnarray*}}
	\newcommand{\enano}{\end{eqnarray*}}
\newcommand{\bee}{\begin{enumerate}}
	\newcommand{\ene}{\end{enumerate}}
\newcommand{\bei}{\begin{itemize}}
	\newcommand{\eni}{\end{itemize}}
\newcommand{\betab}{\begin{tabular}}
	\newcommand{\entab}{\end{tabular}}
\newcommand{\bd}{\begin{displaymath}}

\newcommand{\Iop}{{\mathbb{I}_{V_{\mathbb{H}}^{R}}}}

\newcommand{\blam}{\mbox{\boldmath $\lambda$}}
\newcommand{\bmu}{\mbox{\boldmath $\mu$}}
\newcommand{\btau}{\mbox{\boldmath $\tau$}}

\newcommand{\bfrakq}{\mbox{\boldmath $\mathfrak q$}}

\newcommand{\bfrakx}{\mbox{\boldmath $\mathfrak x$}}
\newcommand{\bfraky}{\mbox{\boldmath $\mathfrak y$}}
\newcommand{\bfraka}{\mbox{\boldmath $\mathfrak a$}}
\newcommand{\bfrakb}{\mbox{\boldmath $\mathfrak b$}}

\newcommand{\bfrakp}{\mbox{\boldmath $\mathfrak p$}}

\newcommand{\bk}{\mathbf k}

\newcommand{\bi}{\mathbf i}
\newcommand{\bj}{\mathbf j}

\begin{document}
\title[Cayley Transform]{S-Spectrum and the quaternionic Cayley transform of an operator}
\author{B. Muraleetharan$^{\dagger}$, I. Sabadini$^*$, K. Thirulogasanthar$^{\ddagger}$}
\address{$^{\dagger}$ Department of mathematics and Statistics, University of Jaffna, Thirunelveli, Sri Lanka.}
\address{$^*$Dipartimento di Matematica, Politecnico di Milano, Via E. Bonardi, 9, 20133, Milano, Italy.}
\address{$^{\ddagger}$ Department of Computer Science and Software Engineering, Concordia University, 1455 De Maisonneuve Blvd. West, Montreal, Quebec, H3G 1M8, Canada.}
\email{bbmuraleetharan@jfn.ac.lk,  irene.sabadini@polimi.it and santhar@gmail.com. }
\date{\today}
\thanks{K. Thirulogasanthar would like to thank the FQRNT, Fonds de la Recherche  Nature et  Technologies (Quebec, Canada) for partial financial support under the grant number 2017-CO-201915. Part of this work was done while he was visiting the Politecnico di Milano to which he expresses his thanks for the hospitality. He also thanks the program Professori Visitatori GNSAGA, INDAM for the support during the period in which this paper was partially written.}
\date{\today}
\begin{abstract}
In this paper we define the quaternionic Cayley transformation of a densely defined, symmetric, quaternionic right linear operator and formulate a general theory of defect number in a right quaternionic Hilbert space. This study investigates the relation between the defect number and S-spectrum, and the properties of the Cayley transform in the quaternionic setting.
\end{abstract}
\maketitle
\pagestyle{myheadings}
\section{Introduction}
Self-adjoint operators play an important role in the Dirac-von Neumann formulation of quantum mechanics.
In complex and in quaternionic quantum mechanics states are described by vectors of a separable complex (resp. quaternionic) Hilbert space and the observables are represented by self-adjoint operators on the respective Hilbert space. By Stone's theorem on one parameter unitary groups, self-adjoint operators are the infinitesimal generators of unitary groups of time evolution.\\

The self-adjointness in a Hilbert space is stronger than being symmetric. Even though the difference is a technical issue, it is very important. For example, the spectral theorem only applies to self-adjoint operators but not to symmetric operators. In this regard, the following question arises in several contexts: if an operator $A$ on a Hilbert space is symmetric, when does it have self-adjoint extensions? In the complex case, an answer is provided by the Cayley transform of a self-adjoint operator and the deficiency indices.\\

Due to the non-commutativity, in the quaternionic case  there are three types of  Hilbert spaces: left, right, and two-sided, depending on how vectors are multiplied by scalars. This fact can entail several problems. For example, when a Hilbert space $\mathcal H$ is one-sided (either left or right) the set of linear operators acting on it does not have a linear structure. Moreover, in a one sided quaternionic Hilbert space, given a linear operator $T$ and a quaternion $\mathfrak{q}\in\quat$, in general we have that $(\mathfrak{q} T)^{\dagger}\not=\overline{\mathfrak{q}} T^{\dagger}$ (see \cite{Mu} for details). These restrictions can severely prevent the generalization  to the quaternionic case of results valid in the complex setting. Even though most of the linear spaces are one-sided, it is possible to introduce a notion of multiplication on both sides by fixing an arbitrary Hilbert basis of $\mathcal H$.  This fact allows to have a linear structure on the set of linear operators, which is a minimal requirement to develop a full theory. Thus, the framework of this paper is a right quaternionic Hilbert space equipped with a left multiplication, introduced by fixing a Hilbert basis. As in the complex case, one may introduce a suitable notion of Cayley type transform of symmetric linear operators. The idea of considering Cayley transform of linear operators is due to von Neumann \cite{VonN} who formally replaced the variable in a Cayley transform by a symmetric operator. The idea was further extended to other types of linear operators but always with the purpose of getting information of the given operator by studying the properties of its Cayley transform. A quaternionic Cayley transform of linear operators appeared in \cite{Vas0, Vas}; however, the type of transform and the underlying notion of spectrum differ from the one treated in this paper.

In this paper, we define the Cayley transform of densely defined symmetric operators  satisfying suitable assumptions. We will prove that this notion of Cayley transform possesses several properties, in particular it is an isometry and allows to prove a characterization of  self-adjointness.

The plan of the paper is the following. The paper consists of four sections, besides the Introduction. In Section 2 we collect some preliminary notations and results on quaternions, quaternionic Hilbert spaces and Hilbert  bases. In Section 3 we introduce right linear operators and some of their properties, the left multiplication, we introduce the notion of deficiency subspace and defect number of an operator at a point also proving some new results in this framework. In Section 4  we study the deficiency indices of isometric operators and we define the notion of quaternionic Cayley transform for  a linear symmetric operator (satisfying suitable hypotheses) and study its main properties. In particular, we show that a linear operator is self-adjoint if and only if its Cayley transform is unitary. In the fifth and last Section, we show that the Cayley transform that we have defined based on the choice of a Hilbert basis, in order to have a left multiplication and thus a two-sided Hilbert space, in fact does not depend on this choice.
\section{Mathematical preliminaries}
In order to make the paper self-contained, we recall some facts about quaternions which may not be well-known.  For details we refer the reader to \cite{Ad,ghimorper,Vis}.
\subsection{Quaternions}
Let $\quat$ denote the field of all quaternions and $\quat^*$ the group (under quaternionic multiplication) of all invertible quaternions. A general quaternion can be written as
$$\bfrakq = q_0 + q_1 \bi + q_2 \bj + q_3 \bk, \qquad q_0 , q_1, q_2, q_3 \in \mathbb R, $$
where $\bi,\bj,\bk$ are the three quaternionic imaginary units, satisfying
$\bi^2 = \bj^2 = \bk^2 = -1$ and $\bi\bj = \bk = -\bj\bi,  \; \bj\bk = \bi = -\bk\bj,
\; \bk\bi = \bj = - \bi\bk$. The quaternionic conjugate of $\bfrakq$ is
$$ \overline{\bfrakq} = q_0 - \bi q_1 - \bj q_2 - \bk q_3 , $$
while $\vert \bfrakq \vert=(\bfrakq \overline{\bfrakq})^{1/2} $ denotes the usual norm of the quaternion $\bfrakq$.
If $\bfrakq$ is non-zero element, it has inverse
$
\bfrakq^{-1} =  \frac {\overline{\bfrakq}}{\vert \bfrakq \vert^2 }.$
Finally, the set
\begin{eqnarray*}
\mathbb{S}&=&\{I=x_1 \bi+x_2\bj+x_3\bk~\vert
~x_1,x_2,x_3\in\mathbb{R},~x_1^2+x_2^2+x_3^2=1\},
\end{eqnarray*}
contains all the elements whose square is $-1$. It is a $2$-dimensional sphere in $\mathbb H$ identified with $\mathbb R^4$.
\subsection{Quaternionic Hilbert spaces}
In this subsection we  discuss right quaternionic Hilbert spaces. For more details we refer the reader to \cite{Ad,ghimorper,Vis}.
\subsubsection{Right quaternionic Hilbert Space}
Let $V_{\quat}^{R}$ be a vector space under right multiplication by quaternions.  For $\phi,\psi,\omega\in V_{\quat}^{R}$ and $\bfrakq\in \quat$, the inner product
$$\langle\cdot\mid\cdot\rangle:V_{\quat}^{R}\times V_{\quat}^{R}\longrightarrow \quat$$
satisfies the following properties
\begin{enumerate}
	\item[(i)]
	$\overline{\langle \phi\mid \psi\rangle}=\langle \psi\mid \phi\rangle$
	\item[(ii)]
	$\|\phi\|^{2}=\langle \phi\mid \phi\rangle>0$ unless $\phi=0$, a real norm
	\item[(iii)]
	$\langle \phi\mid \psi+\omega\rangle=\langle \phi\mid \psi\rangle+\langle \phi\mid \omega\rangle$
	\item[(iv)]
	$\langle \phi\mid \psi\bfrakq\rangle=\langle \phi\mid \psi\rangle\bfrakq$
	\item[(v)]
	$\langle \phi\bfrakq\mid \psi\rangle=\overline{\bfrakq}\langle \phi\mid \psi\rangle$
\end{enumerate}
where $\overline{\bfrakq}$ stands for the quaternionic conjugate. It is always assumed that the
space $V_{\quat}^{R}$ is complete under the norm given above and separable. Then,  together with $\langle\cdot\mid\cdot\rangle$ this defines a right quaternionic Hilbert space. Quaternionic Hilbert spaces share many of the standard properties of complex Hilbert spaces.

The next two Propositions can be established following the proof of their complex counterparts, see e.g. \cite{ghimorper,Vis}.
\begin{proposition}\label{P1}
Let $\mathcal{O}=\{\varphi_{k}\,\mid\,k\in N\}$
be an orthonormal subset of $V_{\quat}^{R}$, where $N$ is a countable index set. Then following conditions are pairwise equivalent:
\begin{itemize}
\item [(a)] The closure of the linear combinations of elements in $\mathcal O$ with coefficients on the right is $V_{\quat}^{R}$.
\item [(b)] For every $\phi,\psi\in V_{\quat}^{R}$, the series $\sum_{k\in N}\langle\phi\mid\varphi_{k}\rangle\langle\varphi_{k}\mid\psi\rangle$ converges absolutely and it holds:
$$\langle\phi\mid\psi\rangle=\sum_{k\in N}\langle\phi\mid\varphi_{k}\rangle\langle\varphi_{k}\mid\psi\rangle.$$
\item [(c)] For every  $\phi\in V_{\quat}^{R}$, it holds:
$$\|\phi\|^{2}=\sum_{k\in N}\mid\langle\varphi_{k}\mid\phi\rangle\mid^{2}.$$
\item [(d)] $\mathcal{O}^{\bot}=\{0\}$.
\end{itemize}
\end{proposition}
\begin{definition}
The set $\mathcal{O}$ as in Proposition \ref{P1} is called a {\em Hilbert basis} of $V_{\quat}^{R}$.
\end{definition}
\begin{proposition}\label{P2}
Every quaternionic Hilbert space $V_{\quat}^{R}$ has a Hilbert basis. All the Hilbert bases of $V_{\quat}^{R}$ have the same cardinality.

Furthermore, if $\mathcal{O}$ is a Hilbert basis of $V_{\quat}^{R}$, then every  $\phi\in V_{\quat}^{R}$ can be uniquely decomposed as follows:
$$\phi=\sum_{k\in N}\varphi_{k}\langle\varphi_{k}\mid\phi\rangle,$$
where the series $\sum_{k\in N}\varphi_k\langle\varphi_{k}\mid\phi\rangle$ converges absolutely in $V_{\quat}^{R}$.
\end{proposition}

It should be noted that once a Hilbert basis is fixed, every left (resp. right) quaternionic Hilbert space also becomes a right (resp. left) quaternionic Hilbert space \cite{ghimorper,Vis}. See next section 3.2 for more details.

The field of quaternions $\quat$ itself can be turned into a left quaternionic Hilbert space by defining the inner product $\langle \bfrakq \mid \bfrakq^\prime \rangle = \bfrakq \overline{\bfrakq^{\prime}}$ or into a right quaternionic Hilbert space with  $\langle \bfrakq \mid \bfrakq^\prime \rangle = \overline{\bfrakq}\bfrakq^\prime$.
\section{Right quaternionic linear  operators and some basic properties}
In this section we shall define right  $\quat$-linear operators and recall some basis properties. Most of them are very well known. In this manuscript, we follow the notations in \cite{AC} and \cite{ghimorper}. We shall also prove some results pertinent to the development of the paper. To the best of our knowledge the results we prove in sections 3.3 and 3.4 do not appear in the literature.
\begin{definition}
A mapping $A:\D(A)\subseteq V_{\quat}^R \longrightarrow V_{\quat}^R$, where $\D(A)$ stands for the domain of $A$, is said to be right $\quat$-linear operator or, for simplicity, right linear operator, if
$$A(\phi\bfraka+\psi\bfrakb)=(A\phi)\bfraka+(A\psi)\bfrakb,~~\mbox{~if~}~~\phi,\,\psi\in \D(A)~~\mbox{~and~}~~\bfraka,\bfrakb\in\quat.$$
\end{definition}
The set of all right linear operators will be denoted by $\mathcal{L}(V_{\quat}^{R})$ and the identity linear operator on $V_{\quat}^{R}$ will be denoted by $\Iop$. For a given $A\in \mathcal{L}(V_{\quat}^{R})$, the range and the kernel will be
\begin{eqnarray*}
\text{ran}(A)&=&\{\psi \in V_{\quat}^{R}~|~A\phi =\psi \quad\text{for}~~\phi \in\D(A)\}\\
\ker(A)&=&\{\phi \in\D(A)~|~A\phi =0\}.
\end{eqnarray*}
We call an operator $A\in \mathcal{L}(V_{\quat}^{R})$ bounded if
\begin{equation*}
\|A\|=\sup_{\|\phi \|=1}\|A\phi \|<\infty,
\end{equation*}
or equivalently, there exist $K\geq 0$ such that $\|A\phi \|\leq K\|\phi \|$ for all $\phi \in\D(A)$. The set of all bounded right linear operators will be denoted by $\B(V_{\quat}^{R})$.
\\
Assume that $V_{\quat}^{R}$ is a right quaternionic Hilbert space, $A$ is a right linear operator acting on it.
Then, there exists a unique linear operator $A^{\dagger}$ such that
\begin{equation}\label{Ad1}
\langle \psi \mid A\phi \rangle=\langle A^{\dagger} \psi \mid\phi \rangle;\quad\text{for all}~~~\phi \in \D (A), \psi\in\D(A^\dagger),
\end{equation}
where the domain $\D(A^\dagger)$ of $A^\dagger$ is defined by
$$
\D(A^\dagger)=\{\psi\in V_{\quat}^{R}\ |\ \exists \varphi\ {\rm such\ that\ } \langle \psi \mid A\phi \rangle=\langle \varphi \mid\phi \rangle\}.
$$
\subsection{Symmetry and Self-adjointness}
Let $A:\D(A)\longrightarrow V_{\quat}^R$ and $B:\D(B)\longrightarrow V_{\quat}^R$ be quaternionic linear operators. As usual, we write $A\subset B$ if $D(A)\subset\D(B)$ and $B\vert_{\D(A)}=A$. In this case, $B$ is said to be an extension of $A$.
\begin{definition}
A right linear operator $A:\D(A)\longrightarrow V_{\quat}^R$ with dense domain is said to be
\begin{itemize}
\item[(a)] \textit{symmetric}, if $A\subset A^{\dagger}$.
\item[(b)] \textit{anti-symmetric}, if $A\subset -A^{\dagger}$.
\item[(c)] \textit{self-adjoint}, if $A= A^{\dagger}$.
\item[(d)] \textit{unitary}, if $\D(A)=V_{\quat}^{R}$ and $A\,A^{\dagger}=A^{\dagger}A=\Iop$.
\item[(e)] \textit{closed}, if the graph $\mathcal{G}:=\D(A)\oplus\text{ran}(A)$ of $A$ is closed in $V_{\quat}^R\times V_{\quat}^R$, equipped with the product topology.
\item[(f)] \textit{closable}, if it admits closed operator extensions. In this case, the \textit{closure} $\overline{A}$ of $A$ is the smallest closed extension and its domain and action are
\begin{itemize}
\item [\textbullet] $\D(\overline{A}):=\{\phi\in V_{\quat}^{R}\,|\, \exists\,\psi\in V_{\quat}^{R}\mbox{~s.t.~} \forall\,\{\phi_n\}\subset\D(A)~\mbox{~with~}\phi_n\rightarrow\phi,\,A\phi_n\rightarrow\psi \}$
\item [\textbullet] $\overline{A}\phi=\psi$.
\end{itemize}
\end{itemize}
\end{definition}
\begin{proposition}\label{copcr}
A right linear operator $A:\D(A)\longrightarrow V_{\quat}^R$ is closed if and only if for any sequence $\{\phi_{n}\}$ in $\D(A)$ such that $\phi_{n}\longrightarrow\phi$ with $A\phi_{n}=\psi_{n}\longrightarrow\psi$ in $V_{\quat}^{R}$, then $\psi=A\phi$.
\end{proposition}
\begin{proof}
It is straightforward from the definition of closed operators.
\end{proof}
\begin{proposition}\label{cldop}
Let $A:\D(A)\longrightarrow V_{\quat}^R$ be densely defined right linear operator. Then
\begin{itemize}
\item[(a)] $A^{\dagger}$ is closed.
\item[(b)] $A$ is closable if and only if $\D(A^{\dagger})$ is dense in $V_{\quat}^R$, and $\overline{A}=A^{\dagger\dagger}$.
\item[(c)] $\text{ran}(A)^{\bot}=\ker(A^{\dagger})$ and $\ker(A)\subset\text{ran}(A^{\dagger})^{\bot}$.\\
Furthermore, if $\D(A^{\dagger})$ is dense in $V_{\quat}^R$ and $A$ is closed, then $\ker(A)=\text{ran}(A^{\dagger})^{\bot}$.
\end{itemize}
\end{proposition}
\begin{proposition}\label{prrpt}
Let $A:\D(A)\subseteq V_{\quat}^R\longrightarrow V_{\quat}^R$ be a right linear operator. If $A$ is closed and satisfies the condition that there exists $C>0$ such that
\be
\|A\phi\|\geq C\|\phi\|,
\label{rpt}
\en
for all $\phi\in D(A)$, then $\text{ran}(A)$ is closed.
\end{proposition}
\begin{proof}
Let $\psi\in\overline{\text{ran}(A)}$, then there exists  a sequence $\{\phi_{n}\}$  in $\D(A)$ such that $A\phi_{n}\longrightarrow \psi$. Then by (\ref{rpt}), we know that
$\{\phi_{n}\}$ is a Cauchy sequence in $V_{\quat}^{R}$ as $\{A\phi_{n}\}$ is Cauchy. Therefore $\phi_{n}\longrightarrow \phi$ for some $\phi\in V_{\quat}^{R}$. From the Proposition (\ref{copcr}), we have $A\phi=\psi$. This completes the proof.
\end{proof}
\begin{proposition}\label{N-S sym}
The right linear operator $A:\D(A)\subseteq V_{\quat}^R\longrightarrow V_{\quat}^R$ is symmetric if and only if $\langle A\phi\mid\phi\rangle\in\mathbb{R}$, for all $\phi\in\D(A)$.
\end{proposition}
\begin{proof}
If $A$ is symmetric, then the statement has been proved in \cite{ghimorper}, Proposition 2.17 (b), but since the proof is short, we repeat it for completeness: for any $\phi\in\D(A)$, $\langle A\phi\mid\phi\rangle=\langle \phi\mid A\phi\rangle=\overline{\langle A\phi\mid\phi\rangle}$. That is, $\langle A\phi\mid\phi\rangle\in\mathbb{R}$, for all $\phi\in\D(A)$. To show the converse, suppose that $\langle A\phi\mid\phi\rangle\in\mathbb{R}$, for all $\phi\in\D(A)$. The polarization identity (see, for example, \cite{ghimorper}) is given by the formula
\begin{equation}\label{polar}
\langle\phi\mid\psi\rangle=\frac{1}{4}(\|\phi+\psi\|-\|\phi-\psi\|)+\frac{1}{4}\sum_{\btau=\bi,\bj,\bk}(\|\phi\btau+\psi\|-\|\phi\btau-\psi\|)\btau,
\end{equation}
where $\phi,\psi\in V_{\quat}^R$. Now using this identity, we can immediately see that
\begin{center}
$\langle A\phi\mid\phi\rangle=\langle \phi\mid A\phi\rangle$, for all $\phi\in\D(A),$
\end{center} that is, $A$ is symmetric. Hence the result follows.
\end{proof}
\subsection{Left Scalar Multiplications on $V_{\quat}^{R}$.}
We shall extract the definition and some properties of left scalar multiples of vectors on $V_{\quat}^R$ from \cite{ghimorper} as needed for the development of the manuscript. The left scalar multiple of vectors on a right quaternionic Hilbert space is an extremely non-canonical operation associated with a choice of preferred Hilbert basis. From the Proposition (\ref{P2}), $V_{\quat}^{R}$ has a Hilbert basis
\begin{equation}\label{b1}
\mathcal{O}=\{\varphi_{k}\,\mid\,k\in N\},
\end{equation}
where $N$ is a countable index set.
The left scalar multiplication on $V_{\quat}^{R}$ induced by $\mathcal{O}$ is defined as the map $\quat\times V_{\quat}^{R}\ni(\bfrakq,\phi)\longmapsto \bfrakq\phi\in V_{\quat}^{R}$ given by
\begin{equation}\label{LPro}
\bfrakq\phi:=\sum_{k\in N}\varphi_{k}\bfrakq\langle \varphi_{k}\mid \phi\rangle,
\end{equation}
for all $(\bfrakq,\phi)\in\quat\times V_{\quat}^{R}$.
\begin{proposition}\cite{ghimorper}\label{lft_mul}
The left product defined in (\ref{LPro}) satisfies the following properties. For every $\phi,\psi\in V_{\quat}^{R}$ and $\bfrakp,\bfrakq\in\quat$,
\begin{itemize}
\item[(a)] $\bfrakq(\phi+\psi)=\bfrakq\phi+\bfrakq\psi$ and $\bfrakq(\phi\bfrakp)=(\bfrakq\phi)\bfrakp$.
\item[(b)] $\|\bfrakq\phi\|=|\bfrakq|\|\phi\|$.
\item[(c)] $\bfrakq(\bfrakp\phi)=(\bfrakq\bfrakp)\phi$.
\item[(d)] $\langle\overline{\bfrakq}\phi\mid\psi\rangle
=\langle\phi\mid\bfrakq\psi\rangle$.
\item[(e)] $r\phi=\phi r$, for all $r\in \mathbb{R}$.
\item[(f)] $\bfrakq\varphi_{k}=\varphi_{k}\bfrakq$, for all $k\in N$.
\end{itemize}
\end{proposition}
\begin{remark}
(1) The meaning of writing $\bfrakp\phi$ is $\bfrakp\cdot\phi$, because the notation from (\ref{LPro}) may be confusing, when $V_{\quat}^{R}=\quat$. However, regarding the field $\quat$ itself as a right $\quat$-Hilbert space, an orthonormal basis $\mathcal{O}$ should consist only of a singleton, say $\{\varphi_{0}\}$, with $\mid\varphi_{0}\mid=1$, because we clearly have $\theta=\varphi_{0}\langle\varphi_{0}\mid\theta\rangle$, for all $\theta\in\quat$. The equality from (f) of Proposition \ref{lft_mul} can be written as $\bfrakp\varphi_{0}=\varphi_{0}\bfrakp$, for all $\bfrakp\in\quat$. In fact, the left hand may be confusing and it should be understood as $\bfrakp\cdot\varphi_{0}$, because the true equality $\bfrakp\varphi_{0}=\varphi_{0}\bfrakp$ would imply that $\varphi_{0}=\pm 1$. For the simplicity, we are writing $\bfrakp\phi$ instead of writing $\bfrakp\cdot\phi$.\\
(2) Also one can trivially see that $(\bfrakp+\bfrakq)\phi=\bfrakp\phi+\bfrakq\phi$, for all $\bfrakp,\bfrakq\in\quat$ and $\phi\in V_{\quat}^{R}$.
\end{remark}
Furthermore, the quaternionic left scalar multiplication of linear operators is also defined in \cite{Fab1}, \cite{ghimorper}. For any fixed $\bfrakq\in\quat$ and a given right linear operator $A:\D(A)\longrightarrow V_{\quat}^R$, the left scalar multiplication of $A$ is defined as a map $\bfrakq A:\D(A)\longrightarrow V_{\quat}^R$ by the setting
\begin{equation}\label{lft_mul-op}
(\bfrakq A)\phi:=\bfrakq (A\phi)=\sum_{k\in N}\varphi_{k}\bfrakq\langle \varphi_{k}\mid A\phi\rangle,
\end{equation}
for all $\phi\in D(A)$. It is straightforward that $\bfrakq A$ is a right linear operator. If $\bfrakq\phi\in \D(A)$, for all $\phi\in \D(A)$, one can define right scalar multiplication of the right linear operator $A:\D(A)\longrightarrow V_{\quat}^R$ as a map $ A\bfrakq:\D(A)\longrightarrow V_{\quat}^R$ by the setting
\begin{equation}\label{rgt_mul-op}
(A\bfrakq )\phi:=A(\bfrakq \phi),
\end{equation}
for all $\phi\in D(A)$. It is also right linear operator. One can easily obtain that, if $\bfrakq\phi\in \D(A)$, for all $\phi\in \D(A)$ and $\D(A)$ is dense in $V_{\quat}^R$, then
\begin{equation}\label{sc_mul_aj-op}
(\bfrakq A)^{\dagger}=A^{\dagger}\overline{\bfrakq}~\mbox{~and~}~
(A\bfrakq)^{\dagger}=\overline{\bfrakq}A^{\dagger}.
\end{equation}
\begin{proposition}\cite{BT}\label{preqn}
Let $A:\D(A)\subseteq V_{\quat}^R\longrightarrow V_{\quat}^R$ be a densely defined right linear symmetric operator with the property that  $\bi\phi,\bj\phi,\bk\phi\in \D(A)$, for all $\phi\in \D(A)$ and $\bfrakq = q_0 + \bi q_1 + \bj q_2 + \bk q_3\in\quat$.
\begin{itemize}
\item [(a)] If $\bi A$, $\bj A$ and $\bk A$ are anti-symmetric, then $(\bfrakq A)^{\dagger}=\overline{\bfrakq}A$, $\bfrakq A=A\bfrakq$ and
	$$\|(A-\bfrakq\Iop)\phi\|^{2}=\|(A-q_0\Iop)\phi\|^{2}+
	(q_1^{2} +  q_2^{2} + q_3^{2})\|\phi\|^{2},$$ for all $\phi\in\D(A)$.
\item [(b)] If $\bfrakq A$ is anti-symmetric, then
$$\|(A- \overline{\bfrakq}\Iop)\phi\|^{2}=\|A\phi\|^{2}+|\bfrakq|^{2}\|\phi\|^{2}$$
for all $\phi\in\D(A)$.
\item [(c)] If $\overline{\bfrakq} A$ is anti-symmetric, then
$$\|(A-\bfrakq\Iop)\phi\|^{2}=\|A\phi\|^{2}+|\bfrakq|^{2}\|\phi\|^{2}$$
for all $\phi\in\D(A)$.
\end{itemize}
\end{proposition}
\begin{proposition}\cite{BT}\label{csadj-gen}
Let $A:\D(A)\subseteq V_{\quat}^R\longrightarrow V_{\quat}^R$ be a densely defined symmetric right linear operator  with the property that $\bi\phi,\bj\phi,\bk\phi\in \D(A)$, for all $\phi\in \D(A)$. If the operators $\bi A$, $\bj A$ and $\bk A$ are anti-symmetric, then for any $\bfrakq = q_0 + \bi q_1 + \bj q_2 + \bk q_3\in\quat$ with $q_1^{2} +  q_2^{2} + q_3^{2}\neq 0$, the following statements are equivalent:
\begin{itemize}
\item[(a)] $A$ is self-adjoint.
\item[(b)] $A$ is closed and $\ker(A^{\dagger}-\bfrakq\Iop)=\{0\}$ and $\ker(A^{\dagger}-\overline{\bfrakq}\Iop)=\{0\}$.
\item[(c)] $\text{ran}(A-\bfrakq\Iop)=\text{ran}(A-\overline{\bfrakq}\Iop)=V_{\quat}^{R}$.
\end{itemize}
\end{proposition}
\subsection{Regular Points and Defect Numbers of quaternionic linear Operators}
\begin{definition}
Let $A:\D(A)\subseteq V_{\quat}^R\longrightarrow V_{\quat}^R$ be a right linear operator. A quaternion $\bfrakq$ is called a \textit{regular point} for $A$ if there exists a number $c_{\bfrakq}>0$ such that
\be\label{reg_pt}
\|(A-\bfrakq\Iop)\phi\|\geq c_{\bfrakq}\|\phi\|,
\en
for all $\phi\in\D(A)$. The set of regular points of $A$ is the \textit{regular domain} of $A$ and it is denoted by $\varPi(A)$.
\end{definition}
\begin{definition}
Let $A:\D(A)\subseteq V_{\quat}^R\longrightarrow V_{\quat}^R$ be a right linear operator. For $\bfrakq\in\varPi(A)$, we call the linear subspace $\text{ran}(A-\bfrakq\Iop)^{\bot}$ of $V_\quat^R$ the \textit{deficiency subspace} of $A$ at $\bfrakq$ and its dimension $d_{\bfrakq}(A)=\dim \text{ran}(A-\bfrakq\Iop)^{\bot}$ is the \textit{defect number} of $A$ at $\bfrakq$.
\end{definition}
The following Proposition discusses some basic properties of the above definitions:
\begin{proposition}\label{pre_set}
Let $A:\D(A)\subseteq V_{\quat}^R\longrightarrow V_{\quat}^R$ be a densely defined right linear operator and $\bfrakq \in\quat$.
\begin{itemize}
\item [(a)] If $\bfrakq_0 \in\varPi(A)$ such that $\mid \bfrakq-\bfrakq_0 \mid<c_{\bfrakq_0}$, where $c_{\bfrakq_0}$ is a constant satisfying (\ref{reg_pt}), then $\bfrakq\in\varPi(A)$. That is, $\varPi(A)$ is an open subset of $\quat$.
\item [(b)] If $A$ is closable, then $\varPi(\overline{A})=\varPi(A)$, $d_{\bfrakq}(\overline{A})=d_{\bfrakq}(A)$ and $$\text{ran}(\overline{A}-\bfrakq\Iop)=\overline{\text{ran}(A-\bfrakq\Iop)},$$ for all $\bfrakq \in\varPi(A)$.
\item [(c)] If $A$ is closed and $\bfrakq \in\varPi(A)$, then $\text{ran}(A-\bfrakq\Iop)$ is a closed linear subspace of $V_{\quat}^R$.
\end{itemize}
\end{proposition}
\begin{proof}
(a) Suppose that $\bfrakq_0 \in\varPi(A)$ such that $\mid \bfrakq-\bfrakq_0 \mid<c_{\bfrakq_0}$ where $c_{\bfrakq_0}$ is a constant satisfying (\ref{reg_pt}). Let $\phi\in\D(A)$, then we have
\begin{eqnarray*}
\|(A-\bfrakq\Iop)\phi\|=\|(A-\bfrakq_0\Iop)\phi-(\bfrakq-\bfrakq_0)\phi\|&\geq&\|(A-\bfrakq_0\Iop)\phi\|-\mid\bfrakq-\bfrakq_0\mid\|\phi\|\\
&\geq&(c_{\bfrakq_{0}}-\mid\bfrakq-\bfrakq_0\mid)\|\phi\|.
\end{eqnarray*}
Thus $\bfrakq\in\varPi(A)$ as $\mid \bfrakq-\bfrakq_0 \mid<c_{\bfrakq_0}$. From this, we can say that $\varPi(A)\subseteq\varPi(A)^{\circ}$-interior set of $\varPi(A)$. Therefore (a) follows.\\
(b) Let $\psi\in\overline{\text{ran}(A-\bfrakq\Iop)}$. Then there exists a sequence $\{\psi_n\}\subset\text{ran}(A-\bfrakq\Iop)$ such that $\psi_n\longrightarrow\psi$ as $n\longrightarrow\infty$. Now for each $n\in\mathbb{N}$, there exists $\phi_n\in\D(A)$ such that
$\psi_n=(A-\bfrakq\Iop)\phi_n$. By (\ref{reg_pt}), we have, for any $m,n\in\mathbb{N}$
$$\|\phi_m-\phi_n\|\leq c_{\bfrakq}^{-1}\|(A-\bfrakq\Iop)(\phi_m-\phi_n)\|=c_{\bfrakq}^{-1}\|\psi_m-\psi_n\|.$$
From this, $\{\phi_n\}$ is a Cauchy sequence, because $\{\psi_n\}$ is a Cauchy sequence. Then $\{\phi_n\}$ is convergent and take $\phi=\lim_n \phi_n $. Since $\lim_n A\phi_n=\lim_n(\psi_n+\bfrakq\phi_n)=\psi+\bfrakq\phi$ and $A$ is closable, we have $\phi\in\D(\overline{A})$ and $\overline{A}\phi=\psi+\bfrakq\phi$. Thus $\psi=(\overline{A}-\bfrakq\Iop )\phi\in\text{ran}(\overline{A}-\bfrakq\Iop)$. This proves that $\overline{\text{ran}(A-\bfrakq\Iop)}\subseteq\text{ran}(\overline{A}-\bfrakq\Iop)$. On the other hand, take $\psi\in\text{ran}(\overline{A}-\bfrakq\Iop)$, then there exists $\phi\in\D(\overline{A})$ such that $\psi=(\overline{A}-\bfrakq\Iop)\phi$. By the definition of $\overline{A}$, we have for all $\{\phi_n\}\subset\D(A)$ with $\phi_n\longrightarrow\phi$, $A\phi_n\longrightarrow \overline{A}\phi$. From this, the converse inclusion follows immediately. Therefore, $$\text{ran}(\overline{A}-\bfrakq\Iop)=\overline{\text{ran}(A-\bfrakq\Iop)}.$$
Let $\bfrakq\in\varPi(A)$, then there exists a number $c_{\bfrakq}>0$ such that
$$\|(A-\bfrakq\Iop)\phi\|\geq c_{\bfrakq}\|\phi\|,$$
for all $\phi\in\D(A)$. Now take $\phi\in\D(\overline{A})$, then
$A\phi_n\longrightarrow \overline{A}\phi$, for all $\{\phi_n\}\subset\D(A)$ with $\phi_n\longrightarrow\phi$, and for each $n\in\mathbb{N}$,
$$\|(A-\bfrakq\Iop)\phi_n\|\geq c_{\bfrakq}\|\phi_n\|.$$
Taking limit $n\longrightarrow\infty$ both side, we get $$\|(\overline{A}-\bfrakq\Iop)\phi\|\geq c_{\bfrakq}\|\phi\|,$$
because $\bfrakq\Iop\phi_n=\sum_{k\in N}\varphi_{k}\bfrakq\langle \varphi_{k}\mid \phi_{n}\rangle\longrightarrow\sum_{k\in N}\varphi_{k}\bfrakq\langle \varphi_{k}\mid \phi\rangle=\bfrakq\Iop\phi$.
This proves that $\varPi(A)\subseteq\varPi(\overline{A})$ and the converse inclusion immediately follows from the fact that $\D(A)\subseteq\D(\overline{A})$.
Thus $\varPi(A)=\varPi(\overline{A})$.\\
Since $\text{ran}(A-\bfrakq\Iop)\subseteq\text{ran}(\overline{A}-\bfrakq\Iop)$, it follows that $\text{ran}(\overline{A}-\bfrakq\Iop)^{\bot}\subseteq\text{ran}(A-\bfrakq\Iop)^{\bot}.$  Let us now take $\psi\in\text{ran}(A-\bfrakq\Iop)^{\bot}$. Let $\xi\in\text{ran}(\overline{A}-\bfrakq\Iop)$, then there exists $\phi\in\D(\overline{A})$ such that $\xi=(\overline{A}-\bfrakq\Iop)\phi$ and
$A\phi_n\longrightarrow \overline{A}\phi$, for all $\{\phi_n\}\subset\D(A)$ with $\phi_n\longrightarrow\phi$. If we take $\xi_n=(A-\bfrakq\Iop)\phi_n$, for all $n\in\mathbb{N}$, then $\langle\psi\mid\xi_n\rangle=0$, for all $n\in\mathbb{N}$. But $\langle\psi\mid\xi_n\rangle\longrightarrow\langle\psi\mid\xi\rangle$. Thus $\langle\psi\mid\xi\rangle=0$. That is, $\psi\in\text{ran}(\overline{A}-\bfrakq\Iop)^{\bot}$. Therefore $\text{ran}(\overline{A}-\bfrakq\Iop)^{\bot}=\text{ran}(A-\bfrakq\Iop)^{\bot}$. Hence $d_{\bfrakq}(\overline{A})=d_{\bfrakq}(A)$ and point (b) follows.\\
(c) Since $A$ is closed, we have $A=\overline{A}$. Thus by the result (b), it immediately follows that $\text{ran}(A-\bfrakq\Iop)$ is a closed subspace. Hence the result (c) follows.
\end{proof}

\subsection{\textsf{S}-Spectrum of Unbounded quaternionic linear Operators}
For a given right linear operator $A:\D(A)\subseteq V_{\quat}^R\longrightarrow V_{\quat}^R$ and $\bfrakq\in\quat$, we define the operator $Q_{\bfrakq}(A):\D(A^{2})\longrightarrow\quat$ by  $$Q_{\bfrakq}(A)=A^{2}-2\text{Re}(\bfrakq)A+|\bfrakq|^{2}\Iop,$$
where $\bfrakq=q_{0}+\bi q_1 + \bj q_2 + \bk q_3$ is a quaternion, $\text{Re}(\bfrakq)=q_{0}$  and $|\bfrakq|^{2}=q_{0}^{2}+q_{1}^{2}+q_{2}^{2}+q_{3}^{2}.$\\
In the literature, the operator $Q_{\bfrakq}(A)$ is sometimes also denoted by $R_{\bfrakq}(A)$ and it is called pseudo-resolvent since it is not the resolvent operator of $A$ but it is the one related to the notion of spectrum as we shall see in the next definition. The notion of $S$-spectrum has been introduced by one of the authors and her collaborators. For more information, the reader may consult e.g. \cite{Fab, Fab1, NFC}, and  \cite{ghimorper}.
\begin{definition}
Let $A:\D(A)\subseteq V_{\quat}^R\longrightarrow V_{\quat}^R$ be a right linear operator. The {\em $S$-resolvent set} (also called \textit{spherical resolvent} set) of $A$ is the set $\rho_{S}(A)\,(\subset\quat)$ such that the three following conditions hold true:
\begin{itemize}
\item[(a)] $\ker(Q_{\bfrakq}(A))=\{0\}$.
\item[(b)] $\text{ran}(Q_{\bfrakq}(A))$ is dense in $V_{\quat}^{R}$.
\item[(c)] $Q_{\bfrakq}(A)^{-1}:\text{ran}(Q_{\bfrakq}(A))\longrightarrow\D(A^{2})$ is bounded.
\end{itemize}
The \textit{$S$-spectrum} (also called \textit{spherical spectrum}) $\sigma_{S}(A)$ of $A$ is defined by setting $\sigma_{S}(A):=\quat\smallsetminus\rho_{S}(A)$. It decomposes into three disjoint subsets as follows:
\begin{itemize}
\item[(i)] the \textit{spherical point spectrum} of $A$: $$\sigma_{pS}(A):=\{\bfrakq\in\quat~\mid~\ker(Q_{\bfrakq}(A))\ne\{0\}\}.$$
\item[(ii)] the \textit{spherical residual spectrum} of $A$: $$\sigma_{rS}(A):=\{\bfrakq\in\quat~\mid~\ker(Q_{\bfrakq}(A))=\{0\},\overline{\text{ran}(Q_{\bfrakq}(A))}\ne V_{\quat}^{R}~\}.$$
\item[(iii)] the \textit{spherical continuous spectrum} of $A$: $$\sigma_{cS}(A):=\{\bfrakq\in\quat~\mid~\ker(Q_{\bfrakq}(A))=\{0\},\overline{\text{ran}(Q_{\bfrakq}(A))}= V_{\quat}^{R}, Q_{\bfrakq}(A)^{-1}\notin\B(V_{\quat}^{R}) ~\}.$$
\end{itemize}
If $A\phi=\phi\bfrakq$ for some $\bfrakq\in\quat$ and $\phi\in V_{\quat}^{R}\smallsetminus\{0\}$, then $\phi$ is called an \textit{eigenvector of $A$ with right eigenvalue} $\bfrakq$. The set of right eigenvalues coincides with the point $S$-spectrum, see \cite{ghimorper}, Proposition 4.5.
\end{definition}
\begin{proposition}
\label{Pr2}\cite{AC, ghimorper} Let $A\in\mathcal{L}(V_H^R)$ and $A$ be self-adjoint, then $\sigma_S(A)\subset\mathbb{R}$.
\end{proposition}
\begin{proposition}\label{pr00}
Let $A:\D(A)\subseteq V_{\quat}^R\longrightarrow V_{\quat}^R$ be right linear operator with the property that $\bi\phi,\bj\phi,\bk\phi\in \D(A)$, for all $\phi\in \D(A)$, and $\bfrakq
\in\quat$. Then the pseudo-resolvent operator $Q_{\bfrakq}(A)$ of $A$ can be written as follows:
\begin{equation}\label{resol}
Q_{\bfrakq}(A)=\frac{1}{2}\left[ (A-\bfrakq\Iop)(A-\overline{\bfrakq}\Iop)+(A-\overline{\bfrakq}\Iop)(A-\bfrakq\Iop)\right].
\end{equation}
Furthermore, if $A$ is densely defined, closed and symmetric with $\overline{\D(A^2)}=V_{\quat}^{R}$, then $\varPi(A)\subseteq\rho_S (A)$.
\end{proposition}
\begin{proof} Formula \eqref{resol} has been proved in \cite{ghirec}, Proposition 5.9. To prove the second part of the statement, suppose that $A$ is a densely defined closed symmetric operator then,  using Cauchy-Schwarz inequality, we have for any $\phi\in\D(A^2)$,
\begin{eqnarray*}
\|Q_{\bfrakq}(A)\phi\|\|\phi\|
&\geq&\langle Q_{\bfrakq}(A)\phi\mid\phi\rangle\\
&=&\frac{1}{2}\langle (A-\bfrakq\Iop)(A-\overline{\bfrakq}\Iop)\phi\mid\phi\rangle+\frac{1}{2}\langle (A-\overline{\bfrakq}\Iop)(A-\bfrakq\Iop)\phi\mid\phi\rangle\\
&=& \frac{1}{2}\|(A-\overline{\bfrakq}\Iop)\phi\|^{2}+
\frac{1}{2}\|(A-\bfrakq\Iop)\phi\|^{2}.
\end{eqnarray*}
That is, for any $\bfrakq\in\quat$ and $\phi\in\D(A^2)$,
\begin{equation}\label{neq1}
\|Q_{\bfrakq}(A)\phi\|\|\phi\|
\geq \frac{1}{2}\|(A-\overline{\bfrakq}\Iop)\phi\|^{2}+
\frac{1}{2}\|(A-\bfrakq\Iop)\phi\|^{2}.
\end{equation}
Thus, if $\bfrakq\in\varPi(A)$, then there exists $c_{\bfrakq}>0$ such that (\ref{reg_pt}) holds. Let $\phi\in\D(A^2)$, then
$$\|Q_{\bfrakq}(A)\phi\|\|\phi\|
\geq\frac{1}{2}\|(A-\bfrakq\Iop)\phi\|^{2}\geq\frac{c_{\bfrakq}^{2}}{2}\|\phi\|^2.$$
This is equivalent to say that for every $\phi\in\D(A^2)$,
\begin{equation}\label{neq2}
\|Q_{\bfrakq}(A)\phi\|\geq\frac{c_{\bfrakq}^{2}}{2}\|\phi\|.
\end{equation}
This inequality (\ref{neq2}) suffices to say that $\ker(Q_{\bfrakq}(A))=\{0\}$ and
$Q_{\bfrakq}(A)^{-1}:\text{ran}(Q_{\bfrakq}(A))\longrightarrow\D(A^{2})$ exists, and is bounded.
Assume that there exists $\psi\in V_{\quat}^{R}\smallsetminus\{0\}$ such that $\psi\bot\,\text{ran}(Q_{\bfrakq}(A))$.
Since $\overline{\D(A^2)}=V_{\quat}^{R}$ and by (c) in Proposition \ref{cldop} we have ${Q_{\bfrakq}(A)}^{\dagger}\psi=0$. Then there exists a sequence $\{\psi_n\}$ in $\D(A^2)$ such that $\psi_n\longrightarrow\psi$. It follows that ${Q_{\bfrakq}(A)}\psi_n={Q_{\bfrakq}(A)}^{\dagger}\psi_n\longrightarrow{Q_{\bfrakq}(A)}^{\dagger}\psi=0$ from ${Q_{\bfrakq}(A)}^{\dagger}$ is closed. That is, ${Q_{\bfrakq}(A)}\psi_n\longrightarrow0$. Since $Q_{\bfrakq}(A)^{-1}$ is continuous, we have
$$\psi_n=Q_{\bfrakq}(A)^{-1}{Q_{\bfrakq}(A)}\psi_n\longrightarrow Q_{\bfrakq}(A)^{-1}0=0.$$
Thus $\psi=0$. This contradicts to the fact that $\psi\ne0$. Therefore $\text{ran}(Q_{\bfrakq}(A))^{\bot}=\{0\}$ and so $\text{ran}(Q_{\bfrakq}(A))$ is dense in $V_{\quat}^{R}$. Hence  $\bfrakq\in\rho_S (A)$, this completes the proof.
\end{proof}
\begin{lemma}\label{dim}\cite{BT}
If $E$ and $F$ are closed linear subspaces of $V_{\quat}^R$ such that $\dim E<\dim F$, then there exists $\psi\in F\cap E^{\bot}$ with $\psi\neq0$.
\end{lemma}
\begin{theorem}\label{def_con}
Let $A:\D(A)\subseteq V_{\quat}^R\longrightarrow V_{\quat}^R$ be a densely defined right linear closable operator and $\bfrakq \in\quat$. Then the defect number $d_{\bfrakq}(A)$ is a constant on each connected component of the open set $\varPi(A)$.
\end{theorem}
\begin{proof}
By the Proposition (\ref{pre_set}), one can assume without loss of generality that $A$ is closed. Then $\text{ran}(A-\bfrakq\Iop)$ is closed by the Proposition (\ref{pre_set}).
Let $\bfrakq_0 \in\varPi(A)$ such that $\mid \bfrakq-\bfrakq_0 \mid<c_{\bfrakq_0}$ where $c_{\bfrakq_0}$ is a constant satisfying (\ref{reg_pt}), then, by (a) of Proposition (\ref{pre_set})$, \bfrakq \in\varPi(A)$.
Assume that $$d_{\bfrakq}(A)=\dim\text{ran}(A-\bfrakq\Iop)^{\bot}<\dim\text{ran}(A-\bfrakq_0\Iop)^{\bot}=d_{\bfrakq_0}(A).$$ Hence by the Lemma (\ref{dim}), we have
there exists $\psi\in \text{ran}(A-\bfrakq_0\Iop)^{\bot}\cap \text{ran}(A-\bfrakq\Iop)$ with $\psi\neq0$. So $\psi=(A-\bfrakq\Iop)\phi$ for some $\phi\in\D(A)$. But, since $\psi\in\text{ran}(A-\bfrakq_0\Iop)^{\bot}$\,, we have
\be\label{ore}
\langle(A-\bfrakq\Iop)\phi\mid(A-\bfrakq_{0}\Iop)\phi\rangle=0
\en
as $(A-\bfrakq_{0}\Iop)\phi\in\text{ran}(A-\bfrakq_0\Iop)$. The same equation (\ref{ore}) can be obtained even when  the inequality $\dim\text{ran}(A-\bfrakq_0\Iop)^{\bot}<\dim\text{ran}(A-\bfrakq\Iop)^{\bot}$ is considered. Now without loss of generality assume that $(A-\bfrakq_0\Iop)\phi\neq0$ and using (\ref{ore}), we can derive, using (b) in the Proposition (\ref{lft_mul}),
\begin{eqnarray*}
\|(A-\bfrakq_0\Iop)\phi\|^{2}
&=&\langle(\bfrakq-\bfrakq_{0})\phi+(A-\bfrakq\Iop)\phi\mid(A-\bfrakq_0\Iop)\phi\rangle\\
&\leq&\mid\bfrakq-\bfrakq_{0}\mid\|\phi\|\|(A-\bfrakq_0\Iop)\phi\|.
\end{eqnarray*}
Thus $\|(A-\bfrakq_0\Iop)\phi\|\leq\mid\bfrakq-\bfrakq_{0}\mid\|\phi\|$. Since $\phi\neq0$ and $\mid \bfrakq-\bfrakq_0 \mid<c_{\bfrakq_0}$, we have
$$\mid\bfrakq-\bfrakq_{0}\mid\|\phi\|<c_{\bfrakq_0}\|\phi\|\leq\|(A-\bfrakq_0\Iop)\phi\|\leq\mid\bfrakq-\bfrakq_{0}\mid\|\phi\|,$$ which is a contradiction. Therefore, for any $\bfrakq_0 \in\varPi(A)$ with $\mid \bfrakq-\bfrakq_0 \mid<c_{\bfrakq_0}$,
$$\dim\ker(A^{\dagger}-\bfrakq_{0}\Iop)=\dim\ker(A^{\dagger}-\bfrakq\Iop)=d_{\bfrakq}(A).$$ Finally let $\bfrakx,\bfraky\in\varPi(A)$ and $\mathscr{P}(\bfrakx,\bfraky)$ be a polygonal path from $\bfrakx$ to $\bfraky$. Now $\{B(\bfrakq,c_{\bfrakq_0})\,:\,\bfrakq\in\mathscr{P}(\bfrakx,\bfraky)\}$ forms a open cover of the compact set $\mathscr{P}(\bfrakx,\bfraky)$, so there exists a finite sub-cover $\{B(\bfrakq_{\tau},c_{\bfrakq_0})\,:\,\tau=1,2,\cdots,s\}$. Since $d_{\bfrakq}(A)$ is constant on each open ball $B(\bfrakq_{\tau},\delta)$, we conclude that $d_{\bfrakx}(A)=d_{\bfraky}(A)$. Hence the theorem follows.
\end{proof}
\begin{proposition}\label{pr01}
Let $A:\D(A)\subseteq V_{\quat}^R\longrightarrow V_{\quat}^R$ be a densely defined right linear closed symmetric operator and $\bfrakq,\overline{\bfrakq}\in\varPi(A)$. Then
$d_{\bfrakq}(A)=d_{\overline{\bfrakq}}(A)=0$ if and only if $A$ is self-adjoint.
\end{proposition}
\begin{proof}
Suppose that $d_{\bfrakq}(A)=d_{\overline{\bfrakq}}(A)=0$. Then by Proposition (\ref{cldop}) we have \begin{equation}\label{eqr2}
\ker(A^{\dagger}-\overline{\bfrakq}\Iop)\times\ker(A^{\dagger}-\bfrakq\Iop)=\text{ran}(A-\bfrakq\Iop)^{\bot}\times\text{ran}(A-\overline{\bfrakq}\Iop)^{\bot}=\{0\}\times\{0\}.
\end{equation}
From (c) in Proposition \ref{pre_set}, one can equivalently write for (\ref{eqr2}) that
\begin{equation}\label{eqr3}
\text{ran}(A-\bfrakq\Iop)\times\text{ran}(A-\overline{\bfrakq}\Iop)={V_{\quat}^{R}}\times{V_{\quat}^{R}}.
\end{equation}
Let $\phi\in\D(A^{\dagger})$, then there exists $\psi\in\D(A)$ such that $(A^{\dagger}-\bfrakq\Iop)\phi=(A-\bfrakq\Iop)\psi$ as $(A^{\dagger}-\bfrakq\Iop)\phi\in V_{\quat}^{R}=\text{ran}(A-\bfrakq\Iop)$. Thus $(A^{\dagger}-\bfrakq\Iop)(\phi-\psi)=0$ which implies that $(\phi-\psi)\in\ker(A^{\dagger}-\bfrakq\Iop)=\{0\}$, and so $\phi=\psi\in\D(A)$. This suffices to say that $A$ is self-adjoint, since $A$ is a symmetric operator.\\
On the other hand, suppose that $A$ is self-adjoint. Since $\bfrakq,\overline{\bfrakq}\in\varPi(A)$, we have there exist $c_{\bfrakq},c_{\overline{\bfrakq}}>0$ such that
\begin{equation}\label{eqr1}
\|(A-\bfrakq\Iop)\phi\|\geq c_{\bfrakq}\|\phi\|\mbox{~~and~~}\|(A-\overline{\bfrakq}\Iop)\phi\|\geq  c_{\overline{\bfrakq}}\|\phi\|,
\end{equation}
for all $\phi\in\D(A)$. From this, we have
$$\text{ran}(A-\bfrakq\Iop)^{\bot}\times\text{ran}(A-\overline{\bfrakq}\Iop)^{\bot}=\ker(A-\overline{\bfrakq}\Iop)\times\ker(A-\bfrakq\Iop)=\{0\}\times\{0\}$$ as $A$ is self-adjoint (i.e. $A=A^\dagger)$. Hence $d_{\bfrakq}(A)=d_{\overline{\bfrakq}}(A)=0$ and this completes the proof.
\end{proof}
\begin{proposition}
Let $A:\D(A)\subseteq V_{\quat}^R\longrightarrow V_{\quat}^R$ be a densely defined right linear closed symmetric operator. Then
$$\{\bfrakq\in\quat~|~\bfrakq,\overline{\bfrakq}\in\varPi(A)\mbox{~~and~~}d_{\bfrakq}(A)=d_{\overline{\bfrakq}}(A)=0\}\subseteq\rho_S (A). $$
\end{proposition}
\begin{proof}
It follows from Propositions \ref{pr00} and \ref{pr01}.
\end{proof}
\section{The quaternionic Cayley Transform}
\begin{proposition}\label{Gen-Von-neq}
Let $A:\D(A)\subseteq V_{\quat}^R\longrightarrow V_{\quat}^R$ be a densely defined right linear symmetric operator with the property that $\bi\phi,\bj\phi,\bk\phi\in \D(A)$, for all $\phi\in \D(A)$. If the operators $\bi A$, $\bj A$ and $\bk A$ are anti-symmetric, then
$$\quat\smallsetminus\mathbb{R}\subseteq\varPi(A).$$
\end{proposition}
\begin{proof}
From (a) in the Proposition \ref{preqn}, we have, for every $\bfrakq\in\quat\smallsetminus\mathbb{R}$, $$\|(A-\bfrakq\Iop)\phi\|\geq
\sqrt{(q_1^{2} +  q_2^{2} + q_3^{2})}\|\phi\|,$$ for all $\phi\in\D(A)$. This suffices to conclude the proof.
\end{proof}
Let $\blam=\lambda_0+\lambda_1\bi+\lambda_2\bj+\lambda_3\bk\in\quat$ be a quaternion such that $\lambda_t>0:\,t=1,2,3$.
\begin{definition}
Let $A:\D(A)\subseteq V_{\quat}^R\longrightarrow V_{\quat}^R$ be a densely defined closed symmetric right linear operator with the property that $\bi\phi,\bj\phi,\bk\phi\in \D(A)$, for all $\phi\in \D(A)$. If the operators $\bi A$, $\bj A$ and $\bk A$ are anti-symmetric then
\begin{equation}\label{def_-}
n(A):=d_{\blam}(A)=\dim\text{ran}(A-\blam\Iop)^{\bot}.
\end{equation}
The number $n(A)$ is called  deficiency index of $A$.
\end{definition}
We note that, in principle, one could have defined two deficiency indices
$n_+(A):=d_{\overline{\blam}}(A)=\dim\text{ran}(A-\overline{\blam}\Iop)^{\bot}$,
$n_-(A):=d_{\blam}(A)=\dim\text{ran}(A-\blam\Iop)^{\bot}$, however Theorem \ref{def_con} implies that $d_{\blam}(A)$ is constant on each connected component of $\varPi(A)$ and thus $n_+(A)=n_-(A)=n(A)$.
\\
Now, we shall verify some elementary facts about linear isometric operators.
\begin{definition}
A right linear operator $U:\D(U)\subseteq V_{\quat}^R\longrightarrow V_{\quat}^R$ is said to be an \textit{isometric operator}, if $\|U\phi\|=\|\phi\|$, for all $\phi\in \D(U)$.
\end{definition}
The following proposition collects some basic aspects of the right linear isometric operators.
\begin{proposition}\label{iso} Let $U:\D(U)\subseteq V_{\quat}^R\longrightarrow V_{\quat}^R$ be  a right linear isometric operator. Then:
\begin{itemize}
\item [(a)] For each $\phi,\psi\in \D (U)$, we have $\langle U\phi\mid U\psi\rangle=\langle\phi\mid\psi\rangle$.
\item [(b)] $U$ is invertible and its inverse $U^{-1}$ is also isometric.
\item [(c)] $U$ is closed if and only if $\D(U)$ is a closed subspace of $V_{\quat}^R$.
\item [(d)] $\quat\smallsetminus\mathbb{S}\subseteq\varPi(U)$.
\end{itemize}
\end{proposition}
\begin{proof}
(a): Since $U$ is linear and isometric, using the polarization identity (\ref{polar}), we can easily obtain the desired relation $\langle U\phi\mid U\psi\rangle=\langle\phi\mid\psi\rangle$, for all $\phi,\psi\in \D(U)$.

(b) and (c): These statements trivially follow.

(d): Since $U$ is an isometry, one can see that
\begin{equation}\label{ine-def}
\|(U-\bmu\Iop)\phi\|\geq\mid\|U\phi\|-\mid\bmu\mid\|\phi\|\mid=\mid1-\mid\bmu\mid\mid\|\phi\|,
\end{equation}
for all $\phi\in \D(U)$ and $\bmu\in\quat\smallsetminus\mathbb S$. This inequality proves (d).
\end{proof}
Theorem \ref{def_con} implies that the defect numbers of the right linear isometric $U:\D(U)\subseteq V_{\quat}^R\longrightarrow V_{\quat}^R$ are constants on the interior of $\mathbb{S}$ and on the exterior of $\mathbb{S}$. The cardinal numbers
\begin{equation}\label{def_int}
d^i(U):=d_{\bmu}(U)=\dim\text{ran}(U-\bmu\Iop)^{\bot}~\mbox{if~}\mid\bmu\mid<1,
\end{equation}
\begin{equation}\label{def_ext}
d^e(U):=d_{\bmu}(U)=\dim\text{ran}(U-\bmu\Iop)^{\bot}~\mbox{if~}\mid\bmu\mid>1
\end{equation}
are called the \textit{deficiency indices} of the isometric operator $U$. The following Lemma gives an interesting result about these numbers.
\begin{lemma} \label{d^i,e}
If $U:\D(U)\subseteq V_{\quat}^R\longrightarrow V_{\quat}^R$ is a right linear isometric operator, then
$d^i(U)=\dim\text{ran}(U)^{\bot}$ and $d^e(U)=\dim\D(U)^{\bot}$.
\end{lemma}
\begin{proof}
If we choose $\bmu=0$ in (\ref{def_int}), then it follows directly that $$d^i(U)=d_0(U)=\dim\text{ran}(U)^{\bot}.$$ Now for any fixed $\bmu\in\quat$ with $0<\mid\bmu\mid<1$, we have
$$(U^{-1}-\bmu\Iop)U\phi=(\Iop-\bmu U)\phi=-\bmu(U-\bmu^{-1}\Iop)\phi,$$ for all $\phi\in \D (U)$. From this, we can obtain $\text{ran}(U^{-1}-\bmu\Iop)=\text{ran}(U-\bmu^{-1}\Iop)$. Therefore by (\ref{def_ext}),
$$d^e(U)=\dim\text{ran}(U-\bmu^{-1}\Iop)^{\bot}=\dim\text{ran}(U^{-1}-\bmu\Iop)^{\bot}.$$
That is $$d^e(U)=d^i(U^{-1})=\dim\text{ran}(U^{-1})^{\bot}=\dim\D(U)^{\bot}$$ follows.
\end{proof}
The following lemma shows an important result.
\begin{lemma}\label{I-U}
If $U:\D(U)\subseteq V_{\quat}^R\longrightarrow V_{\quat}^R$ is a right linear isometric operator and $\text{ran}(\Iop-U)$ is dense in $V_{\quat}^R$, then $\ker(\Iop-U)=\{0\}$.
\end{lemma}
\begin{proof}
If $\phi\in\ker(\Iop-U)$, then $(\Iop-U)\phi=0$ and so $U\phi=\phi$. Thus for any $\psi\in\D(U)$,
$$\langle(\Iop-U)\psi\mid\phi\rangle=\langle\psi\mid\phi\rangle-\langle U\psi,\phi\rangle=\langle\psi\mid\phi\rangle-\langle U\psi,U\phi\rangle=\langle\psi\mid\phi\rangle-\langle \psi,\phi\rangle=0.$$ This is enough to say that $\phi=0$ as $\text{ran}(\Iop-U)$ is dense in $V_{\quat}^R$.
\end{proof}
Let $A:\D(A)\subseteq V_{\quat}^R\longrightarrow V_{\quat}^R$ be a densely defined right linear symmetric operator with the property that $\bi\phi,\bj\phi,\bk\phi\in \D(A)$, for all $\phi\in \D(A)$ and the operators $\bi A$, $\bj A$ and $\bk A$ are anti-symmetric. Since $\lambda_t>0:\,t=1,2,3$ (i.e. $\lambda_t\neq0:\,t=1,2,3$), we have $\overline{\blam}\in\varPi(A)$ by Proposition \ref{Gen-Von-neq}. Hence, $A-\overline{\blam}\Iop$ is invertible. Using these facts, the $\quat$-Cayley transform is defined as follows:
\begin{definition}
Let $A:\D(A)\subseteq V_{\quat}^R\longrightarrow V_{\quat}^R$ be a densely defined right linear symmetric operator with the property that $\bi\phi,\bj\phi,\bk\phi\in \D(A)$, for all $\phi\in \D(A)$ and the operators $\bi A$, $\bj A$ and $\bk A$ are anti-symmetric. The operator $U_A:\D(U_A)\subseteq V_{\quat}^R\longrightarrow V_{\quat}^R$ defined by
\begin{equation}\label{Cay1}
U_A=(A-\blam\Iop)(A-\overline{\blam}\Iop)^{-1}, \mbox{~~with~~} \D(U_A)=\text{ran}(A-\overline{\blam}\Iop),
\end{equation} is said to be the $\quat$-\textit{Cayley transform} of $A$. That is, $U_A$ is defined by
\begin{equation}\label{Cay2}
U_A(A-\overline{\blam}\Iop)\phi=(A-\blam\Iop)\phi, \mbox{~~for all~~} \phi\in\D(A).
\end{equation}
\end{definition}
From now on we use the terminology \textquoteleft Cayley transform\textquoteright \, for simplicity rather than saying $\quat$-Cayley transform. Some useful properties of the Cayley transform are summarized in the following Proposition:
\begin{proposition}\label{Cay_Prn}
Let $A:\D(A)\subseteq V_{\quat}^R\longrightarrow V_{\quat}^R$ be a densely defined right linear symmetric operator with the property that $\bi\phi,\bj\phi,\bk\phi\in \D(A)$, for all $\phi\in \D(A)$ and the operators $\bi A$, $\bj A$ and $\bk A$ are anti-symmetric. If $U_A$ is the Cayley transform of $A$, then the following statements hold:
\begin{itemize}
\item [(a)] The Cayley transform $U_A$ is an isometric operator on $V_{\quat}^R$ with domain $\D(U_A)=\text{ran}(A-\overline{\blam}\Iop)$ and range $\text{ran}(U_A)=\text{ran}(A-\blam\Iop)$.
\item [(b)] $\text{ran}(\Iop-U_A)=\D(A)$ and $A=(\blam\Iop-\overline{\blam}U_A)(\Iop-U_A)^{-1}$.
\item [(c)] $A$ is closed if and only if $U_A$ is closed.
\item [(d)] If $B$ is another densely defined right linear symmetric operator with the property that $\bi\phi,\bj\phi,\bk\phi\in \D(B)$, for all $\phi\in \D(B)$ and the operators $\bi B$, $\bj B$ and $\bk B$ be anti-symmetric, then $A\subseteq B$ if and only if $U_A\subseteq U_B$.
\item [(e)] $d^i(U_A)=d^e(U_A)=n(A)$.
\end{itemize}
\end{proposition}
\begin{proof}
(a): Let $\phi\in\D(A)$. By (a) in the Proposition (\ref{preqn}), we have
$$\|(A-\blam\Iop)\phi\|^{2}=\|(A-\lambda_0\Iop)\phi\|^{2}+
(\lambda_1^{2} +  \lambda_2^{2} + \lambda_3^{2})\|\phi\|^{2}$$
and one can notice that $\|(A-\blam\Iop)\phi\|=\|(A-\overline{\blam}\Iop)\phi\|$. Take $\psi=(A-\overline{\blam}\Iop)\phi$, then $$\|U_A\psi\|=\|(A-\blam\Iop)\phi\|=\|(A-\overline{\blam}\Iop)\phi\|=\|\psi\|.$$ Therefore the Cayley transform $U_A$ is isometric. Now, $\D(U_A)=\text{ran}(A-\overline{\blam}\Iop)$ follows from (\ref{Cay1}) and $\text{ran}(U_A)=\text{ran}(A-\blam\Iop)$ follows from (\ref{Cay2}). \\
(b): Let $\phi\in\D(A)$ and  $\psi=(A-\overline{\blam}\Iop)\phi$. Then one can see that \begin{equation}\label{eq I-U}
(\Iop-U_A)\psi=(\blam-\overline{\blam})\phi,
\end{equation}
and this gives $\text{ran}(\Iop-U_A)=\D(A)$ by recalling that $(\blam-\overline{\blam})\neq0$. Since $\D(A)$ is dense in $V_{\quat}^R$ and by the Lemma \ref{I-U}, we have $\ker(\Iop-U)=\{0\}$. (Observe that this implication could be obtained directly using the equation (\ref{eq I-U}) as well.) Moreover a direct calculation shows that
\begin{equation}\label{eq l_I-U}
(\blam\Iop-\overline{\blam}U_A)\psi=(\blam-\overline{\blam})A\phi,
\end{equation}
where $\psi=(A-\overline{\blam}\Iop)\phi$ and $\phi\in\D(A)$. From (\ref{eq I-U}) and (\ref{eq l_I-U}), we get, for any $\phi\in\D(A)$,
\begin{equation}\label{Aphi}
(\blam\Iop-\overline{\blam}U_A)(\Iop-U_A)^{-1}(\blam-\overline{\blam})\phi=(\blam-\overline{\blam})A\phi.
\end{equation}
Statement (a) in Proposition \ref{preqn} gives $\bfrakq A=A\bfrakq$, and using (\ref{Aphi}) we obtain that
\begin{equation}\label{Aphi1}
(\blam\Iop-\overline{\blam}U_A)(\Iop-U_A)^{-1}(\blam-\overline{\blam})\phi=A((\blam-\overline{\blam})\phi),
\end{equation}
for all $\phi\in\D(A)$. For any $\phi\in\D(A)$, one can choose that $\xi=(\blam-\overline{\blam})^{-1}\phi$ since $(\blam-\overline{\blam})\neq0$. Since $\xi\in\D(A)$, it follows that $$(\blam\Iop-\overline{\blam}U_A)(\Iop-U_A)^{-1}(\blam-\overline{\blam})\xi=A(\blam-\overline{\blam})\xi$$ from equation (\ref{Aphi1}). Thereby, using (c) in Proposition \ref{lft_mul}, we obtain that $$(\blam\Iop-\overline{\blam}U_A)(\Iop-U_A)^{-1}\phi=A\phi$$ and $\D(A)\subseteq\D((\blam\Iop-\overline{\blam}U_A)(\Iop-U_A)^{-1})$. Now if $\phi\in\D((\blam\Iop-\overline{\blam}U_A)(\Iop-U_A)^{-1})$ then $\phi\in\D((\Iop-U_A)^{-1})=\text{ran}(\Iop-U_A)=\D(A)$. Thus $$\D((\blam\Iop-\overline{\blam}U_A)(\Iop-U_A)^{-1})=\D(A).$$ Hence  we have proved that $A=(\blam\Iop-\overline{\blam}U_A)(\Iop-U_A)^{-1}$.\\
(c):  On the one hand, since $\blam,\overline{\blam}\in\varPi(A)$, by (c) in Proposition \ref{pre_set}, it follows that, if $A$ is closed, then $\text{ran}(A-\blam\Iop)$ and $\text{ran}(A-\overline{\blam}\Iop)$ are closed subspaces of $V_{\quat}^{R}$. Thus $\D(U_A)\oplus\text{ran}(U_A)$ is a closed subspace of $V_{\quat}^{R}\times V_{\quat}^{R}$ by above statement (a), and so $U_A$ is closed. On the other hand, if $U_A$ is closed, assume that for a sequence $\{\phi_n\}\subseteq\D(A)$, $\phi_n\longrightarrow\phi$ with $A\phi_n\longrightarrow\phi$. Since $A\phi_n -\overline{\blam}\phi_n\longrightarrow\psi-\overline{\blam}\phi$ and $U_A$ is closed, we have $U_A(A\phi_n -\overline{\blam}\phi_n)\longrightarrow U_A(\psi-\overline{\blam}\phi).$ Using the fact that $U_A$ is a linear isometry, we obtain
$$\|U_A(A\phi_n -\overline{\blam}\phi_n)- U_A(\psi-\overline{\blam}\phi)\|=\|(A\phi_n -\overline{\blam}\phi_n)-(\psi-\overline{\blam}\phi)\|.$$ That is, $(A\phi_n -\overline{\blam}\phi_n)\longrightarrow(\psi-\overline{\blam}\phi)$. Hence $A\phi_n\longrightarrow\psi$ because $\overline{\blam}\phi_n\longrightarrow\overline{\blam}\phi$. Therefore, $A$ is closed.\\
(d): This follows from formula (\ref{Cay2}).\\
(e): By Theorem \ref{def_con}, Lemma \ref{d^i,e} and the above statement (a), one obtains that
 \begin{equation*}
 d^i(U_A)=d_{\blam}(U_A)=\dim\text{ran}(U_A-\blam\Iop)^{\bot}=n(A),
 \end{equation*}
 \begin{equation*}
 d^e(U_A)=d_{\overline{\blam}}(U_A)=\dim\text{ran}(U_A-\overline{\blam}\Iop)^{\bot}=d_{\blam}(U_A)=n(A).
 \end{equation*}
 Hence the statement follows.
\end{proof}
Now we assume that $U:\D(U)\subseteq V_{\quat}^R\longrightarrow V_{\quat}^R$ is a right linear isometric operator and $\text{ran}(\Iop-U)$ is dense in $V_{\quat}^R$. Then by Lemma \ref{I-U}, $\ker(\Iop-U)=\{0\}$ and the operator $\Iop-U$ is invertible. The operator
\begin{equation}\label{Cay_inv1}
A_U=(\blam\Iop-\overline{\blam}U)(\Iop-U)^{-1}\mbox{~~with domain~~}\D(A_U)=\text{ran}(\Iop-U)
\end{equation}
is called the \textit{inverse Cayley transform} of $U$. From this definition we have
\begin{equation}\label{Cay_inv2}
A_U(\Iop-U)\psi=(\blam\Iop-\overline{\blam}U)\psi\mbox{~~for all~~}\psi\in\D(U).
\end{equation}
Next Proposition is a consequence of the above discussion on the inverse Cayley transform.
\begin{proposition}\label{Cay_Prn1}
Let $U:\D(U)\subseteq V_{\quat}^R\longrightarrow V_{\quat}^R$ be a densely defined right linear symmetric and isometric operator with the property that $\bi\psi,\bj\psi,\bk\psi\in \D(U)$, for all $\psi\in \D(U)$ and the operators $\bi U$, $\bj U$ and $\bk U$ are anti-symmetric. Suppose that $\text{ran}(\Iop-U)$ is dense in $V_{\quat}^R$. Then the operator
\begin{equation*}
A_U=(\blam\Iop-\overline{\blam}U)(\Iop-U)^{-1}\mbox{~~with domain~~}\D(A_U)=\text{ran}(\Iop-U)
\end{equation*}
is a densely defined right linear symmetric operator which has Cayley transform $U$.
\end{proposition}
\begin{proof}
 It is immediate that $A_U$ is a right linear operator. Take $\phi\in\D(A_U)$, then $\phi=(\Iop-U)\psi$ for some $\psi\in\D(U)$. Since $U$ is isometric, we  have
\begin{eqnarray*}
\langle A_U\phi\mid\phi\rangle&=&\langle A_U(\Iop-U)\psi\mid(\Iop-U)\psi\rangle\\
&=&\langle (\blam\Iop-\overline{\blam}U)\psi\mid(\Iop-U)\psi\rangle\\
&=& \langle\blam\Iop\psi\mid\Iop\psi\rangle+\langle\overline{\blam}U\psi\mid U\psi\rangle-\langle\blam\Iop\psi\mid U\psi\rangle-\langle\overline{\blam}U\psi\mid\Iop\psi\rangle.
\end{eqnarray*}
By (a) in Proposition (\ref{preqn}), we get
$$\langle A_U\phi\mid\phi\rangle=\langle\blam\Iop\psi\mid\Iop\psi\rangle+\langle U\overline{\blam}\psi\mid U\psi\rangle-\langle\blam\Iop\psi\mid U\psi\rangle-\langle\overline{\blam}U\psi\mid\Iop\psi\rangle.$$
Since $U$ is isometric and using (d) in Proposition \ref{lft_mul}, we have
\begin{eqnarray*}
\langle A_U\phi\mid\phi\rangle&=&\langle\blam\psi\mid\psi\rangle+\langle \psi\mid \blam\psi\rangle-\langle\blam\psi\mid U\psi\rangle-\langle U\psi\mid\blam\psi\rangle\\
&=&2\mbox{\textbf{Re}}[\langle\blam\psi\mid\psi\rangle-\langle\blam\psi\mid U\psi\rangle].
\end{eqnarray*}
Thus $\langle A_U\phi\mid\phi\rangle\in\mathbb{R}$, for all $\phi\in\D(A_U)$. Hence by (\ref{N-S sym}), $A_U$ is symmetric. Since $\D(A_U)=\text{ran}(\Iop-U)$ is dense in $V_{\quat}^R$, $A_U$ is densely defined. Furthermore, for any $\psi\in\D(U)$, using (\ref{Cay_inv2}), we derive the equalities:
$$(A_U-\overline{\blam}\Iop)(\Iop-U)\psi=(\blam-\overline{\blam})\psi\mbox{\,~~and~~\,}(A_U-\blam\Iop)(\Iop-U)\psi=(\blam-\overline{\blam})U\psi.$$
Using point (a) in Proposition (\ref{preqn}), these two equalities imply that
\begin{equation}
U(\blam-\overline{\blam})\psi=(A_U-\blam\Iop)(A_U-\overline{\blam}\Iop)^{-1}(\blam-\overline{\blam})\psi.
\end{equation}
For any $\psi\in\D(A)$, one can choose $\vartheta=(\blam-\overline{\blam})^{-1}\psi$ since $(\blam-\overline{\blam})\neq0$. Since $\vartheta\in\D(A)$, it follows that
$$U(\blam-\overline{\blam})\vartheta=(A_U-\blam\Iop)(A_U-\overline{\blam}\Iop)^{-1}(\blam-\overline{\blam})\vartheta.$$
Thereby, using (c) in Proposition (\ref{lft_mul}), we get the desired formula
$$U\psi=(A_U-\blam\Iop)(A_U-\overline{\blam}\Iop)^{-1}\psi$$
and $\D(U)\subseteq\D((A_U-\blam\Iop)(A_U-\overline{\blam}\Iop)^{-1})$. Since
$$\D((A_U-\blam\Iop)(A_U-\overline{\blam}\Iop)^{-1})\subseteq\D((A_U-\overline{\blam}\Iop)^{-1})=\text{ran}(A_U-\overline{\blam}\Iop)=\D(U)$$ we conclude that $U=(A_U-\blam\Iop)(A_U-\overline{\blam}\Iop)^{-1}$. That is, $U$ is the Cayley transform of $A_U$. Hence the result follows.
\end{proof}
In the complex case, it is enough to assume that $U$ is a linear isometry and that $\text{ran}(\Iop-U)$ is a dense subspace. But in the quaternionic case, we need some more conditions on $U$. Specifically $U$ should be a densely defined right linear symmetric operator with the property that $\bi\psi,\bj\psi,\bk\psi\in \D(U)$, for all $\psi\in \D(U)$ and the operators $\bi U$, $\bj U$ and $\bk U$ should be anti-symmetric.  The following theorem collects the main facts contained in the above two propositions.  In order to state the theorem, we need to introduce the following sets of  right quaternionic linear operators:
$$\mathfrak{X}=\{A~|~\bi\phi,\bj\phi,\bk\phi\in \D(A),~\forall\,\phi\in \D(A)\mbox{~~and~~}(\btau A)^\dagger=-\btau A~\forall\,\btau=\bi,\bj,\bk\},$$
$$\mathfrak{Y}=\{A~|~A\in\mathfrak{X} \mbox{~~is densely defined and symmetric~~}   \}$$
and
$$\mathfrak{Z}=\{U~|~U \mbox{~~is isometric~~}  \mbox{~~and~~} \overline{\text{ran}(\Iop-U)}=V_{\quat}^R \}.$$
\begin{theorem}\label{ess_Cay}
The Cayley transform $\mathscr{C}:\mathfrak{Y}\longrightarrow\mathfrak{Z}$ defined by $$ \mathscr{C}(A)=U_A=(A-\blam\Iop)(A-\overline{\blam}\Iop)^{-1},$$ for all $A\in\mathfrak{X}$, is a injective mapping. Its inverse is the inverse Cayley transform $\mathscr{C}^{-1}:\mathfrak{Y}\cap\mathfrak{Z}\longrightarrow\mathfrak{Y}$, defined by
$$\mathscr{C}^{-1}(U)=A_U=(\blam\Iop-\overline{\blam}U)(\Iop-U)^{-1},$$
for all $U\in\mathfrak{Y}\cap\mathfrak{Z}$.
\end{theorem}
\begin{proof}
The proof can be obtained from the combination (a) and (b) in Proposition \ref{Cay_Prn}, and Proposition \ref{Cay_Prn1}.
\end{proof}
\begin{remark}
To show that the class $\mathfrak{Y}\cap\mathfrak{Z}$ is non-empty and contains non-trivial elements, it is sufficient to consider $2\times 2$ matrices of the form $A_{\theta}=\left[\begin{matrix} \cos\theta & \sin\theta \\ \sin\theta & -\cos\theta\end{matrix}\right]$. They trivially belong to the class $\mathfrak{Y}\cap\mathfrak{Z}$ and they allow to construct examples of matrices of any size still belonging to the same class. In fact it is sufficient to consider matrices of size $2n\times 2n$ of the form ${\rm diag } (A_{\theta_1},\ldots,A_{\theta_n} )$ or matrices of size $2n+1\times 2n+1$ of the form ${\rm diag } (\pm 1, A_{\theta_1},\ldots,A_{\theta_n} )$. One can perform arbitrary permutations of the elements on the diagonal still obtaining matrices in the class $\mathfrak{Y}\cap\mathfrak{Z}$.
\end{remark}
Due to the non-commutativity of the quaternions, $\mathscr{C}$ cannot be bijective despite the fact that the analogous map in the complex case is bijective.\\
 We prove two corollaries of this Theorem.
\begin{corollary}
Let $A\in\mathfrak{Y}$, $U_A$ be its Cayley transform and let $U_A\in\mathfrak{Y}\cap\mathfrak{Z}$. Then $A$ is self-adjoint if and only if  $U_A$ is unitary.
\end{corollary}
\begin{proof}
From the Proposition (\ref{csadj-gen}), $A$ is self-adjoint if and only if $\text{ran}(A-\blam\Iop)=\text{ran}(A-\overline{\blam}\Iop)=V_{\quat}^{R}$, it is equivalent to say $U_A$ is unitary, because of (a) in Proposition (\ref{Cay_Prn}).
\end{proof}
\begin{corollary}\label{cor1}
A unitary operator $U\in\mathfrak{Y}\cap\mathfrak{Z}$ is the Cayley transform of a self-adjoint operator if and only if $\ker({\Iop-U})=\{0\}$.
\end{corollary}
\begin{proof}
By Theorem (\ref{ess_Cay}), there exists $A_U\in\mathfrak{Y}$ such that $\mathscr{C}^{-1}(U)=A_U$. That is, $A_U$ is densely defined. Thus $\overline{\D(A_U)}=\overline{\text{ran}(\Iop-U)}=V_{\quat}^{R}$. From (c) in Proposition (\ref{cldop}), the conclusion follows.
\end{proof}
\section{Partial Invariance of Cayley Transform}
In the previous section we defined the Cayley transform  using a left multiplication defined in terms of a fixed basis of a right quaternionic Hilbert space. However, a natural question arises: whether the defined Cayley transform is invariant under the basis change and we will show, in this section, that the invariance holds partially.

Let $\mathfrak{O}=\{\vartheta_{k}\,\mid\,k\in N\}$ be a Hilbert basis different from the basis in (\ref{b1}) for $V_{\quat}^{R}$. For any given $\bfrakq\in\quat\smallsetminus\{0\}$ define the operators $\mathcal{L}_{\bfrakq}$ and $\mathfrak{L}_{\bfrakq}$  by
\begin{equation}\label{Lcurl}
\mathcal{L}_{\bfrakq}\phi=\bfrakq\cdot\phi=\sum_{k\in N}\varphi_{k}\bfrakq\langle \varphi_{k}\mid \phi\rangle,
 \end{equation}
 and
 \begin{equation}\label{Lfrak}
 \mathfrak{L}_{\bfrakq}\phi=\bfrakq\ast\phi=\sum_{l\in N}\vartheta_{l}\bfrakq\langle \vartheta_{l}\mid \phi\rangle,
 \end{equation}
 for all $\phi\in V_{\quat}^{R}$. We would like to remind to the reader that, up to the end of section 4 what we called $\bfrakq\phi$ is now written as $\bfrakq\cdot\phi$. We wrote it in this new way for individuating it from the other left-scalar-multiplication $\bfrakq\ast\phi$. The following proposition provides some useful results.
\begin{proposition}\label{Pro_lft}
For each $\bfrakq\in\quat$,  $\mathcal{L}_{\bfrakq}=\mathfrak{L}_{\bfrakq}$ if and only if $\langle\varphi\mid\vartheta\rangle\in\mathbb{R}$, for every $(\varphi,\vartheta)\in\mathcal{O}\times\mathfrak{O}$.
\end{proposition}
\begin{proof}
An equivalent form of this statement has been proved in \cite{ghimorper}, Proposition 3.1.
\end{proof}
Next result concludes this section. Note that, in the following Theorem, for each $\btau=\bi,\bj,\bk$, $\btau\phi$ and $\btau A$ denote both $\btau\cdot\phi$ and $\btau\ast\phi$.
\begin{theorem}
Let $A:\D(A)\subseteq V_{\quat}^R\longrightarrow V_{\quat}^R$ be a densely defined right linear symmetric operator with the property that $\bi\phi,\bj\phi,\bk\phi\in \D(A)$, for all $\phi\in \D(A)$ and the operators $\bi A$, $\bj A$ and $\bk A$ are anti-symmetric. Let $U_A$ and $V_A$ be the Cayley transforms of $A$ defined by
$$U_A=(A-\blam\cdot\Iop)(A-\overline{\blam}\cdot\Iop)^{-1}, \mbox{~~with~~} \D(U_A)=\text{ran}(A-\overline{\blam}\cdot\Iop)$$
and
$$V_A=(A-\blam\ast\Iop)(A-\overline{\blam}\ast\Iop)^{-1}, \mbox{~~with~~} \D(V_A)=\text{ran}(A-\overline{\blam}\ast\Iop)$$
respectively. If $\langle\varphi\mid\vartheta\rangle\in\mathbb{R}$, for every $(\varphi,\vartheta)\in\mathcal{O}\times\mathfrak{O}$, then $U_A=V_A$.
\end{theorem}
\begin{proof}
It is straightforward from the Proposition (\ref{Pro_lft}).
\end{proof}


\begin{thebibliography}{XXXX}
\bibitem{Ad} Adler, S.L., {\em Quaternionic Quantum Mechanics and Quantum Fields}, Oxford University Press, New York, 1995.

\bibitem{AC} Alpay, D., Colombo, F., Kimsey, D.P., {\em The spectral theorem for quaternionic unbounded normal operators based on the $S$-spectrum}, J. Math. Phys. {\bf 57} (2016), 023503.

\bibitem{Fab} Colombo, F., Sabadini, I., \textit{On Some Properties of the Quaternionic Functional Calculus}, J. Geom. Anal., {\bf 19} (2009), 601-627.

\bibitem{Fab1} Colombo, F., Sabadini, I., \textit{On the  Formulations of the Quaternionic Functional Calculus}, J. Geom. Phys., {\bf 60} (2010), 1490-1508.



\bibitem{NFC} Colombo, F., Sabadini, I., Struppa, D.C., {\em Noncommutative Functional Calculus}, Birkh\"auser Basel, 2011.

\bibitem{DeV} DeVito, C.L., \textit{Functional analysis and linear operator theory}, Allan M. Wylde, 1990.	

\bibitem{ghimorper} Ghiloni, R., Moretti, W. and Perotti, A., {\em Continuous slice functional calculus in quaternionic Hilbert spaces\/,} Rev. Math. Phys. {\bf 25} (2013), 1350006.

\bibitem{ghirec} Ghiloni, R., Recupero, V., {\em Slice regular semigroups},  arXiv:1605.01645.

\bibitem{BT} Muraleetharan, B., Thirulogasanthar, K., {\em Deficiency Indices of Some Classes of Unbounded $\quat$-Operators},  Complex Anal. Oper. Theory (2017), 1-29.

\bibitem{Mu} Muraleetharan, B, Thirulogasanthar, K., {\em Coherent state quantization of quaternions}, J. Math. Phys., {\bf 56} (2015), 083510.

    \bibitem{Vas0} Vasilescu , F.-H., {\em
Quaternionic Cayley transform},
J. Funct. Anal. {\bf 164} (1999),  134-162.

\bibitem{Vas} Vasilescu , F.-H., {\em Quaternionic Cayley transform revisited}, J. Math. Anal. Appl., {\bf 409} (2014), 790-807.



\bibitem{Vis} Viswanath, K., {\em Normal operators on quaternionic Hilbert spaces}, Trans. Amer. Math. Soc. {\bf 162} (1971), 337-350.

\bibitem{VonN} von Neumann J., {\em Allgemeine Eigenwerttheorie Hermitescher Funktionaloperatoren}, Math. Ann. {\bf 102} (1929-1930), 49-131.

\end{thebibliography}
\end{document}